\documentclass[draft]{tbimc}
\usepackage{amsmath}
\usepackage{amsfonts}
\usepackage{dsfont}
\usepackage{verbatim}

\newtheorem{thm}{Theorem}[section]
\newtheorem{lem}{Lemma}[section]
\newtheorem{nas}{Corollary}[section]
\theoremstyle{remark}
\newtheorem{zau}{Remark}[section]
\theoremstyle{definition}
\newtheorem{ozn}{Definition}[section]

\DeclareMathOperator{\M}{\mathbb E}

\DeclareMathOperator{\Dom}{Dom}

\newcommand{\si}{\sigma}
\newcommand{\wf}{\widetilde{f}}

\newcommand{\ws}{\widetilde{\sigma}}
\newcommand{\os}{\overline{\sigma}}

\begin{document}


\title[Option pricing under stochastic volatility]
	{Application of Malliavin calculus to exact and approximate option pricing under stochastic volatility}

\author{S. Kuchuk-Iatsenko}
\address{Taras Shevchenko National University of Kyiv}
\email{kuchuk.iatsenko@gmail.com}	
	
\author{Y. Mishura}
\address{Taras Shevchenko National University of Kyiv}
\email{myus@univ.kiev.ua}


\author{Y. Munchak}
\address{Taras Shevchenko National University of Kyiv}
\email{yevheniamunchak@gmail.com}


\selectlanguage{english}
\begin{abstract}
The article is devoted to models of financial markets with stochastic volatility, which is defined by a functional of Ornstein-Uhlenbeck process or Cox-Ingersoll-Ross process. We study the question of exact price of European option. The form of the density function of the random variable, which expresses the average of the volatility over time to maturity is established using Malliavin calculus.The result allows calculate the price of the option with respect to minimum martingale measure when the Wiener process driving the evolution of asset price and the Wiener process, which defines volatility, are uncorrelated.
\end{abstract}

\maketitle

\section{Introduction}
An exact and approximate option pricing in the models with stochastic volatility has been extensively studied during last decades. There are many factors which stimulate such investigations, and probably the major of them are desire to enhance the classic Black--Scholes model and rapid increase in computational resources. Among the fundamental works in this field one might outline \cite{Heston1993}, \cite{HullWhite}, \cite{Steins}, \cite{Wiggins}. Authors of the above papers consider price of an option as a solution of partial differential equation with respect to two variables, price of an asset and volatility, which was derived in
\cite{Garman}. In
\cite{HullWhite} the approximate price of European option is derived in series form taking into account the distribution of an asset price conditional on average value of volatility. Authors of
\cite{Heston1993} and \cite{Steins} determine analytical formulas for the price of European option by application of inverse Fourier transform for the case when there is no correlation between the asset price process and the volatility process. On the contrary to these approaches the finite difference method is applied in \cite{Wiggins} in order to solve the aforementioned differential equation, which allows to investigate the problem in the most general setting.

The results obtained in the aforementioned works have paved the way for new studies. Thus, the inverse Fourier transform approaches in different variations still remain a widely used tool in determination of analytical formulas for option prices in various models. Among the later works dealing with the matters of derivation of exact and approximate formulas for option prices one might outline \cite{NicolatoVenardos}, where an analytical expression for the price of an option is derived for the class of non-Gaussian models with stochastic volatility driven by Ornstein--Uhlenbeck process (see also \cite{BarndorffShephard}). One might also mention the results obtained in the works  \cite{DIppolitiMoretto}, \cite{Goard}, \cite{SKI_YSM-1}, \cite{PerSirMas}, where authors determine exact or approximate formulas of option prices in various models. Thus, in \cite{DIppolitiMoretto} the diffusion models with jumps are considered for which the inverse Fourier transform is applied to derive the analytical expression for the option price. In \cite{PerSirMas} the Hermite polynomials techniques are deployed in order to obtain the formula for approximate option price for the model in which the price of an asset is given by geometric Brownian motion and the volatility is exponential function of Ornstein--Uhlenbeck process. The similar model is considered in \cite{SKI_YSM-1}: volatility is assumed to be some function of Ornstein--Uhlenbeck process. Under assumption that the asset price process and volatility process are uncorrelated the inverse Fourier transform is applied to obtain the analytical formula for European option price.  In \cite{SKI_YSM-2} the Euler--Maruyama discretization approach to derivation of approximate option price for the similar model is suggested and the rate of convergence of approximate price to the true price is estimated. Authors of \cite{mish-munch-2} investigate the rate of convergence of option prices when the asset prices in discrete-time models converge weakly to the Black--Scholes price. The discrete approximation scheme of the asset prices modeled by the geometric Ornstein--Uhlenbeck process is considered and rate of convergence of fair option prices is derived in \cite{mish-munch-3}. In \cite{Goard} the author applies Lie symmetry methods to the problem of solution of the aforementioned partial differential equations in the Heston 3/2 model, that is the model in which the volatility is the solution of stochastic differential equation
$dY_t=Y_t(y-\alpha Y_t)dt+kY^{3/2}_tdW_t.$ One might find great deal of information about the investigations of financial market models with stochastic volatility, for example in \cite{ShephardAndersen}.

The main problem about the exact pricing of an option is that the option price is a function of integral functional argument which depends on the trajectory of the volatility process. The distribution of this integral functional is generally unknown. However it appears that Malliavin calculus techniques may be applied to determine the probability density function of integral functional of stochastic volatility, and this is what we study in this paper.

Application of Malliavin calculus to financial mathematics has spread widely after the approach to construction of optimal portfolios of assets by means of Clark--Ocone formula was presented in \cite{OconeKaratzas}. In \cite{Fournie}  Malliavin calculus framework is used to derive formulas for so called ``Greeks'' -- the quantities representing the sensitivity of the price of options to a change in underlying parameters on which the value of an instrument or portfolio of financial instruments is dependent.
Such application of Malliavin calculus to financial mathematics remains one of the most popular. However other applications are also developed (see \cite{Nualart}, \cite{Sanz},
\cite{NunnoOksProske} and references therein).

This paper is structured as follows: Section 2 introduces general setting of the Black--Scholes models with stochastic volatility driven by Ornstein--Uhlenbeck or Cox--Ingersoll--Ross processes. The results concerning no arbitrage properties of the models and expressions for the price of European option are presented.
Some fundamental concepts of Malliavin calculus along with preliminary results are covered in Section 3 followed by the main result of this work -- the theorem about probability density function of the distribution of average volatility. The expression for the price of European option in terms of this probability density function is presented thereafter.  Section 4 contains proofs of some auxiliary results, such as, expressions of stochastic derivatives of functionals of stochastic volatility.

\section{Some properties of the Black--Scholes model with stochastic volatility driven by some function of Ornstein--Uhlenbeck or Cox--Ingersoll--Ross processes. Expression for the price of European option}

\subsection{Some properties of the Black--Scholes model with stochastic volatility driven by a function of Ornstein--Uhlenbeck process}

Let $\{\Omega, \mathcal{F}, \mathbf{F}=\{\mathcal{F}_t^{(W,\widetilde{W})},$ $t\geq 0\}, \mathbb{P}\}$  be complete probability space with filtration generated by Wiener processes $\{W_t,$ $\widetilde{W}_t,$ $0 \leq t \leq T\}.$ We consider the model of the market where one risky asset and one risk-free asset, and the price of the latter at the moment of time $t$ is given by  $B_t=\operatorname{e}^{rt},$ where $r>0.$  The price of risky asset evolves according to the geometric Brownian motion $\{S_t,\;0 \leq t \leq T\} $ and its volatility is given by a measurable function of another stochastic process. In this paper we assume that the latter process is either Ornstein--Uhlenbeck process or Cox--Ingersoll--Ross process. This subsection deals with the first case. More precisely, the market is described by the pair of stochastic differential equations, the first of which is linear with respect to the asset price and the second one is of Langevin type:
\begin{equation}\label{ModelA0}
dS_t   =   \mu S_tdt+\sigma(Y_t)S_tdW_t,
\end{equation}
\begin{equation}\label{ModelA1}
dY_t   =  -\alpha Y_tdt+kd\widetilde{W}_t.
\end{equation}

Denote $S_0 $ and $Y_0$ deterministic initial values of the processes specified by equations \eqref{ModelA0}--\eqref{ModelA1}, respectively. Let
 $$\overline{X}_t=(1,X_t)=(1,e^{-rt}S_t)$$ be a vector of discounted prices of assets.

In order to reduce the level of technical complexity of the model we impose the following restrictions:

\begin{itemize}

\item [(A1)]  Wiener processes   $W$ and $\widetilde{W}$ are uncorrelated, and hence, independent;

\item [(A2)] the volatility function $\sigma : \mathbb{R}\rightarrow \mathbb{R}_+$   is measurable, bounded away from zero by a constant and has at most polynomial growth, that is
    $c \leq \sigma(x) \leq q(1+|x|^l)$ for all  $x\in\mathbb{R}$
    and some positive constants $c, q$ and  $l\in \mathds{N}$.

\item [(A3)] the coefficients $\alpha  $ and $k$  are positive.

\end{itemize}

The solution of equation \eqref{ModelA0} is of the form
\begin{equation*}
S_t=S_0\exp{\left(\mu t - \frac{1}{2}\int^t_0\sigma^2(Y_s)ds+\int^t_0\sigma(Y_s)dW_s\right)}.
\end{equation*}

Corresponding discounted asset is
\begin{equation*}
X_t=S_0\exp{\left((\mu-r) t - \frac{1}{2}\int^t_0\sigma^2(Y_s)ds+\int^t_0\sigma(Y_s)dW_s\right)}
\end{equation*}
which satisfies the following equation:
\begin{equation}\label{discount}
dX_t=  (\mu-r)dt +   \sigma(Y_s)X_sdW_s.
\end{equation}
Equation \eqref{discount} yields that the discounted price process can be represented by $X_t=S_0+M_t+A_t, $ where $M_t=\int^t_0\sigma(Y_s)X_sdW_s$ is a continuous local martingale,
$A_t=(\mu-\sigma)t$ is a continuous process with finite variation.
  The Ornstein--Uhlenbeck process given by \eqref{ModelA1} is a convenient tool for the purpose of modeling volatility on financial markets due to its mean-reversion property. This process is Gasussian with the following characteristics:
\begin{eqnarray*}
\mathbb{E}[Y_t]=Y_0\operatorname{e}^{-\alpha t}, \quad \operatorname{Var}[Y_t]=\frac{k^2}{2\alpha}(1-\operatorname{e}^{-2\alpha t}).
\end{eqnarray*}

Moreover, the Ornstein--Uhlenbeck process is Markov and admits the explicit representation:

\begin{equation*}
\label{OU}
Y_t=Y_0 \operatorname{e}^{-\alpha t}+k\int_{0}^t \operatorname{e}^{-\alpha (t-s)}d\widetilde{W}_s.
\end{equation*}

\subsection{Some properties of the Black--Scholes model with stochastic volatility driven by Cox--Ingersoll--Ross process}
Now we consider the model of the market where one risky asset and one risk-free asset, and the price of the latter at the moment of time $t$ is given by
$B_t=\operatorname{e}^{rt},$ where $r>0.$ The price of risky asset is given by the following pair of stochastic differential equations
\begin{equation}\label{ModelB0}
dS_t   =   \mu S_tdt+\sqrt{Z_t}S_tdW_t,
\end{equation}
\begin{equation}\label{ModelB1}
dZ_t=(b-Z_t)dt+k\sqrt{Z_t}d\widetilde{W}_t.
\end{equation}
Denote $Z_0> 0 $   deterministic initial values of the process given by \eqref{ModelB1}.
Let condition $(A1)$ hold along with the following additional condition:
\begin{itemize}

\item [(A3')] coefficients $b$  and $k$  are positive and $k^2< 2b$.

\end{itemize}
The solution of equation \eqref{ModelB0} is of the form
$$
S_t=S_0\exp{\left(\mu t - \frac{1}{2}\int^t_0Z_sds+\int^t_0 \sqrt{Z_s}dW_s\right)}.
$$
Cox--Ingersoll--Ross process given by \eqref{ModelB1} has the following characteristics
$$\mathbb{E}[Z_t]=Z_0 \operatorname{e}^{-t}+b(1- \operatorname{e}^{-t}),$$
$$\operatorname{Var}[Z_t]=Z_0 k^2( \operatorname{e}^{-t}- \operatorname{e}^{-2t})+\frac{bk^2}{2}\left(1- \operatorname{e}^{-t}\right)^2.$$

According to \cite{Cox_Ing_Ross} the condition $k^2< 2b$ is necessary and sufficient for the process $Z$ to attain positive values and not to hit zero. Throughout the paper we suppose this condition to be fulfilled. Model \eqref{ModelB0}--\eqref{ModelB1} is called the Heston model.

\subsection{Absence of arbitrage, incompleteness and equivalent martingale measures in the model with stochastic volatility driven by a function of Ornstein--Uhlenbeck process}

The question of absence of arbitrage in model \eqref{ModelA0}--\eqref{ModelA1} is crucial for the problem of option pricing. It is investigated in detail in  \cite{SKI_YSM-1},
and only key results are mentioned here. It is well-known that there are several definitions of absence of arbitrage for semimartingale models in continuous time. They are covered in detail in \cite{DelbSchach} and \cite{Shiryaev} and differ by, for example, classes of admissible trading strategies. We consider the notion of absence of arbitrage in $\overline{{NA}_{g}}$ sense (\cite{DelbSchach,   Shiryaev}),  that is for the case when the class of admissible trading strategies consists of such self-financing strategies that maximum loss or debt over the portfolio at any moment of time $t \in [0,T]$ is bounded from below by the following scalar product:
$(\overline{g}, \overline{X}_t),$ where $\overline{g}$ is some vector with positive components and $\overline{X}_t$ is a vector of discounted prices of assets traded on the market. Naturally the absence of arbitrage is connected with existence of martingale measures.
\begin{ozn}\label{EMMeas}
A probability measure $\mathbb{Q},$ which is equivalent to the objective measure $\mathbb{P},$ is called an equivalent martingale measure if the discounted price process is a martingale under the measure $\mathbb{Q}.$
\end{ozn}
According to theorem 2, \cite[p. 653]{Shiryaev} the existence of a martingale measure yields the absence of arbitrage of our model in $\overline{{NA}_{g}}$ sense.

Hence, due to classical Girsanov theorem the set of martingale measures is the subset of the set of measures which have Radon--Nikodym derivative restriction on $\mathcal{F}_t$ of the following form:
\begin{equation}
\label{Lt}
\begin{gathered}
\frac{d\mathbb{Q}}{d\mathbb{P}}\Big|_{\mathcal{F}_t}= \exp \Big( \int_0^t (r-\mu)/\sigma(Y_s)dW_s+\int_0^t \nu_sd\widetilde{W}_s \\
-\frac{1}{2}\int_0^t ((r-\mu)^2/\sigma^2(Y_s)+\nu^2_s)ds\Big),
\end{gathered}
\end{equation}
where    $\nu=(\nu_t)_{0\leq t \leq T}$ is progressively measurable process for which  $\int^T_0
\nu_s^2ds<\infty$ $\mathbb{P}-$a.s.. It is obvious that under condition (A2) of boundedness away from zero of volatility and under assumption that the process  $\nu$ is bounded
all set functions $\mathbb{Q}$ having Radon--Nikodym derivative of the form  \eqref{Lt} define martingale measures.
 As there are more than one martingale measure the market is incomplete.
The pair of processes $(S_t, Y_t)$ have the following representation with respect to the equivalent martingale measure $\mathbb{Q}$ with Radon--Nikodym derivative  \eqref{Lt}:
\begin{equation*}\label{ModelQ}
\begin{array}{rcl}
dS_t & = & r S_tdt+\sigma(Y_t)S_tdW^\mathbb{Q}_t, \\
dY_t & = & \left( -\alpha Y_t-k \nu(t)\right) dt+kd\widetilde{W}^\mathbb{Q}_t,
\end{array}
\end{equation*}
where according to two-dimensional Girsanov theorem (see, for example, Theorem 5.4.1, \cite{Shreve}) the processes
\begin{align*}
W^\mathbb{Q}_t & = W_t+\int_0^t\dfrac{\mu-r}{\sigma(Y_s)}ds, \\
\widetilde{W}^\mathbb{Q}_t & = \widetilde{W}_t+\int_0^t\nu(s)ds,
\end{align*}
are independent Wiener processes with respect to $\mathbb{Q}$.

Obviously among all measures given by  \eqref{Lt} the simplest form has the one having $\nu(s)\equiv 0$. At the same time according to Theorem 5.1 from
\cite{SKI_YSM-1} this measure is minimal martingale measure in the sense of the following definition.

\begin{ozn} Let a price of discounted asset on a financial market is $\mathbb{P}$-semimartingale  $X$ which has a representation $X=X_0+M+A,$ where $M$ is local
$\mathbb{P}$-martingale, $A$ is an adapted process with finite variation. A martingale measure $\mathbb{Q}$ which is equivalent to the objective measure $\mathbb{P},$ is called a minimal martingale measure (MMM) if $\mathbb{Q}=\mathbb{P}$ on $\mathcal{F}_0,$ and any square-integrable $\mathbb{P}$-martingale strictly orthogonal to the process $M$, is a local $\mathbb{Q}$-martingale.
\end{ozn}

Notice that according to what is given above the components of decomposition in our model are $X_0=S_0,\  M_t=\int_0^t \sigma(Y_s)X_sdW_s,\  A_t= (\mu-r)t.$
Below we will study the option prices with respect to minimal martingale measure. With respect to such measure (denote it  $\mathbb{Q}$) equations
\eqref{ModelA0}--\eqref{ModelA1} gain the following form (see. Section 5, \cite{SKI_YSM-1}):
\begin{equation}\label{ModelB}
\begin{array}{rcl}
dS_t & = & r S_tdt+\sigma(Y_t)S_tdW^\mathbb{Q}_t, \\
dY_t & = &-\alpha Y_tdt+kd\widetilde{W}^\mathbb{Q}_t,
\end{array}
\end{equation}
where stochastic processes
\begin{align*}
W^\mathbb{Q}_t & = W_t+\int_0^t\dfrac{\mu-r}{\sigma(Y_s)}ds, \\
\widetilde{W}^\mathbb{Q}_t & = \widetilde{W}_t,
\end{align*}
are independent Wiener processes with respect to measure $\mathbb{Q}.$

\subsection{Absence of arbitrage, incompleteness and equivalent martingale measures in the model with stochastic volatility driven by Cox--Ingersoll--Ross process}

  In this model the set of martingale measures is the subset of the set of measures which have Radon--Nikodym derivative restriction on $\mathcal{F}_t$ of the following form:
\begin{equation}
\label{Lt_1}
\begin{gathered}
\frac{d\mathbb{Q}}{d\mathbb{P}}\Big|_{\mathcal{F}_t}= \exp\left\{\int\limits_0^t\frac{r-\mu}{\sqrt{Z_s}}dW_s+\int\limits_0^t\nu_{1,s}d\widetilde{W}_s
 -\frac{1}{2}\int\limits_0^t\left( \frac{(r-\mu)^2}{Z_s}+\nu_{1,s}^2\right)ds\right\},
\end{gathered}
\end{equation}
where    $\nu_1=(\nu_{1,t})_{0\leq t \leq T}$ is a progressively measurable process for which  $\int^T_0\nu_{1,s}^2ds<\infty$ $\mathbb{P}-$a.s..  In order to prove the absence of arbitrage on the market set $\nu_{1,s}=0$. Then \eqref{Lt_1} yields
\begin{equation}
\label{Lt_11}
\frac{d\mathbb{Q}}{d\mathbb{P}}\Big|_{\mathcal{F}_t}= L_{1,t}:=
\exp\left\{\int\limits_0^t\frac{r-\mu}{\sqrt{Z_s}}dW_s-\frac{1}{2}\int\limits_0^t\left(\frac{(r-\mu)^2}{Z_s}\right)ds\right\},
\end{equation}
and according to Theorem 3.6 and Corollary 3.3 in \cite{Wong} we have  $\M[L_{1,t}]=1$ and the discounted price process
$$X_t=\exp\left\{\int\limits_0^t\sqrt{Z_s}d\widetilde{W}_s^{\mathbb{Q}}-\frac{1}{2}\int\limits_0^tZ_sds\right\}$$
 is $\mathbb{Q}$-martingale and market has no-arbitrage property.The pair of processes $(S_t, Y_t)$ have the following representation with respect to the equivalent martingale measure $\mathbb{Q}$ with Radon--Nikodym derivative \eqref{Lt_11}:
\begin{equation}\label{ModelZ}
\begin{array}{rcl}
dS_t & = & r S_tdt+\sqrt{Z_t}S_tdW^\mathbb{Q}_t, \\
dZ_t & = & \left( b-Z_t\right) dt+k\sqrt{Z_t}d\widetilde{W}^\mathbb{Q}_t,
\end{array}
\end{equation}
where according to two-dimensional Girsanov theorem processes
\begin{align*}
W^\mathbb{Q}_t & = W_t+\int_0^t\dfrac{\mu-r}{\sqrt{Z_s}}ds, \\
\widetilde{W}^\mathbb{Q}_t & = \widetilde{W}_t,
\end{align*}
are independent Wiener processes with respect to measure $\mathbb{Q}$. Similarly to the previous subsection this measure is minimal martingale measure.

\subsection{European option price as a function of volatility in the model with sto\-- chastic volatility}\label{OpPrice}

Denote $V_C$ a price at the initial moment of time of European call option $C=(S_T-K)^+$ with strike price $K\geq 0$ in model \eqref{ModelB}. This price is given by the following expression:
\begin{equation}\label{EV_0}
V_C=\operatorname{e}^{-rT}\mathbb{E}^{\mathbb{Q}}\{(S^{\mathbb{Q}}_T-K)^{+}\}=
\operatorname{e}^{-rT}\mathbb{E}^{\mathbb{Q}}\{\mathbb{E}^{\mathbb{Q}}\{(S^{\mathbb{Q}}_T-K)^{+}|Y_s, 0 \leq s\leq T\}\}.
\end{equation}

The inner expectation is conditional with respect to the trajectory  $\{Y_s, 0 \leq s \leq T\},$ and thus is the Black--Scholes price in the model with deterministic time-dependent volatility. According to Lemma 2.1 in \cite{paper4} the inner expectation in \eqref{EV_0}, denote it $E(\bar{\sigma})$,  has the following representation:
\begin{gather}\label{black-scholes}
\notag
E(\bar{\sigma}):=\mathbb{E}^{\mathbb{Q}}\{(S^{\mathbb{Q}}_T-K)^{+}|Y_s, 0 \leq s\leq T\}\\
\label{IntExp}
=S_0\operatorname{e}^{rT}\Phi\left(\frac{\ln S_0+(r+\frac{1}{2}\bar{\sigma}^2)T-\ln K}{\bar{\sigma}\sqrt{T}}\right)\\
\notag
-K\Phi\left(\frac{\ln S_0+(r-\frac{1}{2}\bar{\sigma}^2)T-\ln K}{\bar{\sigma}\sqrt{T}}\right),
\end{gather}
where $\bar{\sigma}:=\left(\dfrac{1}{T}\int_0^T\sigma^2(Y_s)ds\right)^{\frac{1}{2}},\;\Phi(\cdot)$ is cumulative distribution function of standard normal distribution. The function
$\bar{\sigma}$ may be viewed as an averaged volatility for the period of time from initial moment until maturity.
Formula \eqref{IntExp} evidences that the option price in Black--Scholes model with stochastic volatility is completely determined by the distribution of random variable $\bar{\sigma}$.

Similarly, the price of European call option  in model \eqref{ModelZ} can be derived by replacing $\bar{\sigma}$ with random variable $\widetilde{\sigma}:=\left(\dfrac{1}{T}\int_0^T Z_sds\right)^{\frac{1}{2}}$.

\section{Stochastic derivative and option price} Now we apply Malliavin calculus and particularly the notion of stochastic derivative to find an expression for probability density functions of the random variables
$\bar{\sigma}$ and $\widetilde{\sigma}$.
\subsection{Malliavin calculus. Probability density function of a random variable as a functional of stochastic derivative}
We begin by recalling necessary definitions and stating the proposition about probability density function of random variable being a functional of stochastic derivative. The fundamentals and applications of Malliavin calculus are given in more detail in \cite{Nualart}.

Let   $W=\{W(t), t \in [0,T]\},$ be a Wiener process on a probability space $\{\Omega, \mathcal{F},$  $\mathbf{F}=\{\mathcal{F}_t^{W}, t \in [0,T],
\mathbb{P}\},$ where $\Omega=C([0,T], \mathbb{R}).$

Denote $\widehat{C}^{\infty} \mathbb(R)$ a set of all infinitely differentiable functions having derivatives of at most polynomial growth on infinity.

\begin{ozn}
Smooth random variable is a random variable $F$ of the form $F=f(W(t_1),$ $\ldots , W(t_n)),$ $f=f(x^1,\ldots,x^n) \in \widehat{C}^{\infty}(\mathbb{R}^n),$ $ t_1, \ldots
t_n \in [0,T].$ We denote by $\mathcal{S}$ the class of smooth random variables.
\end{ozn}

\begin{ozn}
Let $F \in \mathcal{S}.$ Stochastic derivative of a random variable $F$ at point  $t$ is the following random variable:
\begin{equation*}
D_tF= \sum_{i=1}^n \frac{\partial f}{\partial x^i}(W(t_1), \ldots , W(t_n)) 1_{[0,t_i]}(t), \quad t\in[0,T].
\end{equation*}
\end{ozn}

The domain of the derivative operator $D:L^2(\Omega) \rightarrow L^2([0,T], \mathbb{R})$ is a Hilbert space  $\mathbb{D}^{1,2}$ of random variables. The scalar product on  $\mathbb{D}^{1,2}$ is defined as follows:
$$
\langle F, G \rangle _{1,2}=\mathbb{E}(FG)+\mathbb{E}(\langle DF, DG \rangle _H), \quad H=L^2[0,T].
$$

The space $\mathbb{D}^{1,2}$ is a dense subset of $L^2(\Omega)$ and a closure of the class of smooth random variables $\mathcal{S}$ with respect to the norm
$$
||F||_{1,2}=[\mathbb{E}(|F|^2)+\mathbb{E}(||DF||^2_H)]^{1/2}.
$$

Hence the derivative operator $D$ is closable, unbounded and is defined on a dense subset of the space $L^2(\Omega)$ (see \cite{Nualart}).

\begin{ozn}\label{ISDef}
Denote by $\delta$ the adjoint of the operator $D$ which is unbounded operator in $L^2([0,T], \mathbb{R})$ with values in $L^2(\Omega)$ such that:
\begin{itemize}
\item[$(i)$] the domain of $\delta$ is the set of square-integrable random variables $u\in L^2([0,T], \mathbb{R})$ such that
$$
\left|\mathbb{E}\left(\left\langle DF, u\right\rangle_H \right)\right|\le C(\mathbb{E}( F^2))^{1/2},
$$
for all $F\in\mathbb{D}^{1,2}$, where $C$ is some constant depending on $u$;
\item[$(ii)$] if $u$ belongs to the domain of $\delta$, then  $\delta(u)$ is the element of $L^2(\Omega)$ characterized by
$$
\mathbb{E}\left(F\delta(u)\right)=\mathbb{E}\left( \left\langle DF, u\right\rangle_H\right)
$$
for any $F\in\mathbb{D}^{1,2}$.
\end{itemize}
\end{ozn}
The operator $\delta$ is closed as the adjoint of an unbounded and densely defined operator. Denote its domain by $\Dom\delta.$

Consider the space $L^{1,2}=L^2([0,T],\mathbb{D}^{1,2})$ with the norm $||\cdot||_{L^{1,2}},$ where
\begin{equation*}
||u||^2_{L^{1,2}}=\mathbb{E}\left(\int_0^T u^2_tdt+\int_0^T\int_0^T \left(D_s u_t \right)^2 dtds\right).
\end{equation*}
\begin{zau}\label{ISL12}
If $u \in L^{1,2},$ then the integral $\delta(u)$ is well defined and the following inequality holds:
$$\mathbb{E}\left(\int_0^T u_t dW_t\right)^2\leq ||u||^2_{L^{1,2}}$$
(see \cite{NuaPar},\cite{Ouknine}). In this case the operator  $\delta(u)$ is called a Skorohod integral of the process $u$ and is denoted by
$$\delta(u)=\int\limits_0^T u_t dW_t.$$
\end{zau}

 The following proposition is crucial for the proof of main result of this paper.
\begin{lem}\textbf{(Proposition 2.1.1 from \cite{Nualart})}
Let $F$ be a random variable from $\mathbb{D}^{1,2}$. Assume that $\frac{DF}{\left\|DF\right\|_H^2}$ belomgs to the domain of the operator $ {\delta}.$ The the probability density function of the random variable $F$ is continuous, bounded and admits the following representation:
$$p(x)=\M \left[1_{\left\{F>x\right\}}\delta\left(\frac{DF}{\left\|DF\right\|_H^2}\right)\right].
$$
\end{lem}

 We will also make use of the following variation of Fubini's theorem for the case of Skorohod integral.
\begin{lem}\label{NuaLeon210}\textbf{(Lemma 2.10 from \cite{NuaLeon})} Let the following conditions hold:
\begin{itemize}
\item [1)] Function $u(t,h, \omega) \in L^2\left([0,T]^2\times \Omega \right)$ and for almost all $t \in [0,T]$ the stochastic process $u(t, \cdot) \in \Dom
    \delta;$
\item [2)]$\mathbb{E}\left[\int_0^T |\delta (u(t, \cdot))|^2 dt\right]<\infty.$
\end{itemize}
Then $\{\int_0^T u(t, h)dt, \ h\in[0,T]\} \in \Dom \delta$ and
\begin{equation*}
\int_0^T\int_0^T u(t,h) dt dW_h=\int_0^T\int_0^T u(t,h) dW_h dt.
\end{equation*}
\end{lem}

 Before we state and prove the main result of this work we recall that as it has been mentioned earlier in section \ref{OpPrice}, the random variables $\bar{\sigma}=\left(\dfrac{1}{T}\int_0^T\sigma^2(Y_s)ds\right)^{\frac{1}{2}}$ and $\widetilde{\sigma}=\left(\dfrac{1}{T}\int_0^T Z_sds\right)^{\frac{1}{2}}$ completely define the option prices in models \eqref{ModelB} and \eqref{ModelZ} respectively.

We introduce the notation
$\nu(x)=\sigma(x)\sigma'(x)$.

\begin{thm}\label{TheoremDensity}
\begin{itemize}
\item[1)] Let the function $\si$ satisfies assumption (A2) and is twice continuously differentiable,   its derivative  $ \si'$ is strictly positive and is of at most polynomial growth on infinity.  Then for the Ornstein--Uhlenbeck process $Y$ defined by stochastic differential equation \eqref{ModelA1} the random variable
   $\os^2$ has continuous bounded probability density function of the form
\begin{equation}\label{DensOU}
p_{\os^2}{(x)}=\mathbb{E}\left[1_{\{\overline{\si} >\sqrt{x}\}}\left(\int_0^T\eta_{t} \int_0^t\operatorname{e}^{\alpha s} dW_s\, dt-\int_0^T\int_0^t
\operatorname{e}^{\alpha h} D_h\eta_{t}\, dh\,dt\right)\right],
\end{equation}
where
\begin{equation}\label{eta}
 \eta_{t}={\alpha T\over k}\operatorname{e}^{-\alpha t} \nu(Y_t)
 \left[ \int\limits_{0}^T \int\limits_{0}^T\Big[\operatorname{e}^{-\alpha |t_1-t_2|}-\operatorname{e}^{-\alpha (t_1+t_2)}\Big]\nu(Y_{t_1}) \nu(Y_{t_2})
 \,dt_1dt_2\right]^{-1},
\end{equation}
  the expression for stochastic derivative $D_h\eta_{t}$ is given in Lemma \ref{Dheta} and all the components of the right-hand sides of equalities \eqref{DensOU} and \eqref{eta}
  are well defined.

\item[2)] Assume $6k^2<b$. For the Cox--Ingersoll--Ross process $Z$ defined by stochastic differential equation \eqref{ModelB1} the random variable $\ws^2$  has continuous bounded probability density function of the form:
\begin{equation}\begin{split}\label{DensKIR}
p_{\ws^2}(x)=&\mathbb{E}\left[1_{\{\ws>\sqrt{x}\}}\left(\frac{T}{k}\int_0^T\sqrt{Z_t}\int_0^t\Psi_{h,t}dW_h
\,dt-\frac{T}{2}\int_0^T\int_0^t\Psi_{h,t}\psi_{h,t}dh\,dt\right)\right],
\end{split}\end{equation}
where
\begin{equation*}\label{psi}
\psi_{h,t}:=\exp\left\{-\frac{t-h}{2}-\left(\frac{b}{2}-\frac{k^2}{8}\right)\int_h^t\frac{ds}{Z_s}\right\},\end{equation*}
\begin{equation*}\label{Psi}
\Psi_{h,t}=\psi_{h,t}\left[\int_{0}^T\int_{0}^T \sqrt{Z_{t_1}}\sqrt{Z_{t_2}}\int_0^{t_1\wedge
t_2}\psi_{h,t_1}\psi_{h,t_2}dh\,dt_1dt_2\right]^{-1}.\end{equation*}
\end{itemize}
\end{thm}
\begin{proof}
$1)$ The stochastic derivative of the Ornstein--Uhlenbeck process has the form
$$
D_hY_t=\left(-\alpha \int_{h}^tD_hY_s\, ds+k\right) 1_{\{h<t\}}.
$$
Solving this equation with respect to $t$ with $h$ fixed we get
$$
D_hY_t=k \operatorname{e}^{-\alpha(t-h)}1_{\{h<t\}}.
$$
Then the stochastic derivative for the bounded continuously differentiable function $\si$ such that $ \si'$ is of at most polynomial growth on infinity is derived by the chain rule:
$$
D_h\si^2(Y_t)=2\si(Y_t)\si'(Y_t)D_hY_t=2k \operatorname{e}^{-\alpha (t-h)}\nu(Y_t)1_{\{h<t\}},
$$
and the stochastic derivative of the integral functional
$$
I_T(\si^2)=\int_0^T\si^2(Y_t)\, dt
$$
 is given by
\begin{equation}\label{difer}
D_hI_T(\si^2)=\int_0^TD_h\si^2(Y_t)\, dt=2k\int_h^T  \operatorname{e}^{-\alpha(t-h)}\nu(Y_t)\,dt, \quad h\leq T.
\end{equation}
Now it is necessary to determine if the following Skorohod integral exists:
$$
\overline{\delta} :=\delta\left(D\os^2\over \|D\os^2\|_H^2\right),
$$
where $\os^2= \frac{I_T(\si^2)}{T}. $

In order to do that we present the process which is being integrated in explicit form. According to \eqref{difer}
$$
D_h\os^2=\frac{2k}{T} \int_h^T  \operatorname{e}^{-\alpha (t-h)}\nu(Y_t)\,dt,
$$
hence,
\begin{gather*}
\|D\os^2\|_H^2=\int_{0}^T (D_h\os^2)^2\, dh=\frac{4k^2}{T^2}\int_0^T\left(\int_h^T   \operatorname{e}^{-\alpha(t-h)}\nu(Y_t)\,dt\right)^2\, dh
\\=\frac{4k^2}{T^2}\int_{0}^T\int_{0}^T\int_0^{t_1\wedge t_2} \operatorname{e}^{-\alpha(t_1-h)}\nu(Y_{t_1})
\operatorname{e}^{-\alpha(t_2-h)}\nu(Y_{t_2})\,dhdt_1dt_2
\\=\frac{2k^2}{\alpha T^2} \int_{0}^T\int_{0}^T\Big[\operatorname{e}^{-\alpha |t_1-t_2|}-\operatorname{e}^{-\alpha (t_1+t_2)}\Big]\nu(Y_{t_1}) \nu(Y_{t_2})
\,dt_1dt_2.
\end{gather*}

Thus, the stochastic process
$$
\overline{\zeta}_h:={D_h\os^2\over \|D\os^2\|^2}
$$
is given by
\begin{equation*}\label{Zeta}
\overline{\zeta}_h=\operatorname{e}^{\alpha h}\int_h^T \eta_{t}\, dt=\int_0^T u(t,h) dt, \quad u(t,h):=\eta_{t}\operatorname{e}^{\alpha h}1_{\{h<t\}}.
\end{equation*}
Lemma \ref{ZetaIntSkor} yields that  $\overline{\zeta}_h \in L^{1,2}$, hence,   the process  $\overline{\zeta}_h$ is Skorohod integrable. Or equivalently, there exists the integral $\overline{\delta}.$ Now we need to check if the conditions of Lemma
\ref{NuaLeon210} are fulfilled for $u(t,h).$ First, we make use of inequalities \eqref{est21} and  \eqref{Ref23} from Lemma   \ref{IntSkor} to derive that
\begin{equation*}
\mathbb{E}\left[\int_0^T \int_0^T u^2(t,h)dhdt\right]\leq \frac{\operatorname{e}^{2\alpha T}-1}{2\alpha}\mathbb{E}\left[\int_0^T\eta^2_t dt\right]< \infty,
\end{equation*}
and consequently $u(t,h, \omega) \in L^2\left([0,T]^2\times \Omega \right).$ Taking into account \eqref{est21}--\eqref{est23} for all $t\in[0,T]$ we get
\begin{equation}\label{SkorIso}
\begin{gathered}
 \mathbb{E}\left[\int_0^T u^2(t,h)dh+\int_0^T\int_0^T (D_s u(t,h))^2dsdh\right] \\
 = \mathbb{E}\left[\int_0^T (\eta_{t}\operatorname{e}^{\alpha h}1_{\{h<t\}})^2dh+\int_0^T\int_0^T (\operatorname{e}^{\alpha h}1_{\{h<t\}} D_s\eta_t)^2dsdh\right]\\
 \leq \frac{\operatorname{e}^{2\alpha T}-1}{2\alpha}\left(\mathbb{E}\eta^2_t
+\int_0^T \mathbb{E}\left[(D_s \eta_t)^2\right]ds\right)\leq C.
\end{gathered}
\end{equation}
  This yields the Skorohod integrability of $u(t,\cdot) $ for all $t\in[0,T]$. Hence, the first condition of Lemma \ref{NuaLeon210} is satisfied.
  Second condition is also fulfilled because $$\mathbb{E}[(\delta (u(t,\cdot)))^2]\leq \mathbb{E}\left[\int_0^T u^2(t,h)dh+\int_0^T\int_0^T (D_s u(t,h))^2dsdh\right]\leq C,$$ which yields that $\mathbb{E}\left[\int_0^T(\delta (u(t,h)))^2dt\right]< \infty.$
Then we can apply Fubini's theorem to the integral
\begin{equation*}
\overline{\delta}=\int_0^T\int_0^T \operatorname{e}^{\alpha h}\eta_{t}1_{\{h<t\}}\, dt dW_h
\end{equation*}
and change the order of integration:

 \begin{equation*}
\overline{\delta}=\int_0^T\int_0^T\operatorname{e}^{\alpha h}\eta_{t}1_{\{h<t\}}\, dW_h  dt.
\end{equation*}

From the last equality and from Theorem 3.2  \cite{NuaPar}  we deduce that
\begin{gather*}
\overline{\delta}=\int_0^T\left(\eta_{t} \int_0^T \operatorname{e}^{\alpha h} 1_{\{h<t\}}\, dW_h-\int_0^T \operatorname{e}^{\alpha h}
D_h\eta_{t}1_{\{h<t\}}dh\right)\,
dt \\
=\int_0^T\eta_{t} \left(\int_0^t \operatorname{e}^{\alpha h} dW_h\right) \,dt-\int_0^T\int_0^t \operatorname{e}^{\alpha h} D_h\eta_{t}\, dhdt.
\end{gather*}

$2)$ Now we will determine the probability density function of the random variable $\ws^2$ in a similar manner. By the Corollary 4.2 from \cite{Alos} the stochastic derivative of the process \eqref{ModelB1} is given by

$$
D_hZ_t=k \exp\left\{\int_h^t\left[-\frac{1}{2}-\left(\frac{b}{2}-\frac{k^2}{8}\right)\frac{1}{Z_s}\right]ds\right\}\sqrt{Z_t}
$$
$$=k \exp\left\{-\frac{t-h}{2}-\left(\frac{b}{2}-\frac{k^2}{8}\right)\int_h^t\frac{ds}{Z_s}\right\}\sqrt{Z_t}=k\psi_{h,t}\sqrt{Z_t}.$$
Then for
$$
I_T(Z_t)=\int_0^TZ_t\, dt
$$
the corresponding stochastic derivative is
$$
D_hI_T(Z_t)=\int_0^TD_h Z_t \, dt=k\int\limits_h^T\psi_{h,t}\sqrt{Z_t}\,dt, \,\,\, h\le T.$$
Now it is necessary to determine if the following Skorohod integral exists
$$
\widetilde{\delta}:=\delta\left(D\ws^2\over \|D\ws^2\|^2\right).
$$
Notice that
\begin{gather*}
D_h\ws^2=\frac{k}{T}\int_h^T  \psi_{h,t}\sqrt{Z_t}\,dt,
\end{gather*} and the corresponding norm equals to
\begin{gather*}
\|D\ws^2\|^2 =\int_{0}^T(D_h\ws)^2\, dh
 =\frac{k^2}{T^2}\int_0^T\left(\int_h^T  \psi_{h,t}\sqrt{Z_t}\,dt\right)^2\, dh
\\=
\frac{k^2}{T^2}\int_{0}^T\int_{0}^T \sqrt{Z_{t_1}}\sqrt{Z_{t_2}}\int_0^{t_1\wedge t_2}\psi_{h,t_1}\psi_{h,t_2}dhdt_1dt_2.
\end{gather*}
Thus the stochastic process
$$
 \widetilde{\zeta}_h:={D_h\ws^2\over \|D\ws^2\|^2}
$$
is given by
$$
\widetilde{\zeta}_h=\frac{T}{k}\int_h^T \sqrt{Z_t}\Psi_{h,t}dt=\frac{T}{k}\int_h^T\widetilde{u}(t,h)dt,\,\,\, \widetilde{u}(t,h):=\sqrt{Z_t}\Psi_{h,t}.
$$

According to Lemma \ref{lem_norma} the process $\widetilde{\zeta}_h$ is Skorohod integrable. Hence, the integral $\widetilde{\delta}$ exists.
It is time to check the conditions of Lemma \ref{NuaLeon210}. By Lemma  \ref{lem_norma}
$$\mathbb{E}\left[\int\limits_0^T\int\limits_0^T\widetilde{u}^2(t,h)dtdh\right]=\int\limits_0^T\int\limits_0^T\mathbb{E}(Z_t\Psi_{h,t}^2)dhdt<\infty$$
 and so $\widetilde{u}(t,h,\omega)\in L^2\left([0,T]^2\times \Omega \right)$. Taking into account inequalities \eqref{sup1}, \eqref{sup2} and applying similar reasoning as in \eqref{SkorIso} it is straightforward to derive the fact that for each $t\in [0,T]$ fixed
 $\widetilde{u}(t,h)\in \Dom \delta.$
 Then the first condition of Lemma \ref{NuaLeon210} is fulfilled. Obviously, the second condition is satisfied too because
\begin{equation*}\begin{gathered}
\mathbb{E}\left[\int_0^T\left(\delta (\widetilde{u}(t,h))\right)^2dt\right]\\
=\int_0^T\mathbb{E}\left[\int_0^T\widetilde{u}^2(t,h)dt+\int_0^T\int_0^TD_s\widetilde{u}(t,h)D_t\widetilde{u}(s,h)dsdt\right]<\infty,
\end{gathered}\end{equation*}
which uses the fact that the expression inside the integral is finite by Lemma \ref{lem_norma}.

We apply consequently Fubini's theorem and Theorem 3.2 from \cite{NuaPar} to get

\begin{gather*}
\widetilde{\delta}=\frac{T}{k}\int_0^T\int_0^T \sqrt{Z_t}\Psi_{h,t}1_{\{h<t\}}dtdW_h\\
=\frac{T}{k}\int_0^T\int_0^T \sqrt{Z_t}\Psi_{h,t}1_{\{h<t\}}dW_hdt\\
=\frac{T}{k}\int_0^T\left(\sqrt{Z_t}\int_0^T\Psi_{h,t}1_{\{h<t\}}dW_h-\int_0^T\Psi_{h,t} D_h\sqrt{Z_t}1_{\{h<t\}}dh\right)dt\\
=\frac{T}{k}\int_0^T\sqrt{Z_t}\int_0^t\Psi_{h,t}dW_h dt-\int_0^T\int_0^t\Psi_{h,t}\frac{D_hZ_t}{2\sqrt{Z_t}}dhdt\\
=\frac{T}{k}\int_0^T\sqrt{Z_t}\int_0^t\Psi_{h,t}dW_h dt-\frac{T}{2}\int_0^T\int_0^t\Psi_{h,t}\psi_{h,t}dhdt.
\end{gather*}
\end{proof}

\begin{nas}
Let the conditions of Theorem \ref{TheoremDensity} hold. Then the price of European call option $C=(S_T-K)^+$ with strike price $K\geq 0$   at the initial moment of time is given by
\begin{gather*}
V_C=\int_0^{\infty} \Bigg(S_0\Phi\left(\frac{\ln S_0+(r+\frac{1}{2}x)T-\ln K}{\sqrt{xT}}\right)\\ -K\operatorname{e}^{-rT}\Phi\left(\frac{\ln
S_0+(r-\frac{1}{2}x)T-\ln K}{\sqrt{xT}}\right)\Bigg)p{(x)} dx,
\end{gather*}
where $$
p{(x)}=\begin{cases}
p_{\os^2}{(x)},&\text{for model \eqref{ModelB},}\,\,\, p_{\os^2}{(x)} \text{is defined by \eqref{DensOU};}\\
p_{\ws^2}(x),&\text{for model \eqref{ModelZ},}\,\,\, p_{\ws^2}{(x)} \text{is defined by \eqref{DensKIR}.}
\end{cases}
$$
\end{nas}

\section{Auxiliary results}
First, we are going to prove the result stating the boundedness of negative order moments of the random variables representing the first moment of time when the Ornstein--Uhlenbeck process $Y$ or Cox--Ingersoll-Ross process $Z$ leave certain interval. For each $Y_0\in \mathbb{R}$ fixed we consider arbitrary interval such that for each $x\in[Y_0-a, Y_0+a]$ the inequality $|\sigma'(x)-\sigma'(Y_0)|\leq \frac{\sigma'(Y_0)}{2}$ holds. Introduce the notation $\tau=\inf\{t>0: |Y_t-Y_0|\geq a\}$, $\tau_1=\tau\wedge T$.  For the Cox--Ingersoll--Ross process starting from the point  $Z_0>0$ we denote $\widetilde{\tau}=\inf\{t>0: |Z_t-Z_0|\geq \frac{Z_0}{2}\}$ and $\widetilde{\tau}_1=\widetilde{\tau}\wedge T$.

 \begin{lem}\label{lemlem} The negative moments of any order of the aforementioned random variables are finite, that is $\mathbb{E}(\tau_1)^{-p}<\infty$ and $\mathbb{E}(\widetilde{\tau}_1)^{-p}<\infty$ for each $p>0 $.
\end{lem}
\begin{proof} According to Lemma 10.5 \cite{Watanabe}, if  $K>0$ and  $X=\{X_t,t\ge 0\}$ is a one-dimensional continuous semimartingale of the form
$$X_t=X_0+M_t+A_t,$$
where $\langle M\rangle_t=\int_0^t\alpha (s)ds$ and  $A_t=\int_0^t\beta (s)ds$ with $|\alpha (s)|\le K$and  $|\beta (s)|\leq K$, then for each $a>0$ and $\lambda\in(0,\frac{a}{2K}]$ the following inequality $$\mathbb{ P}\{\tau_a<\lambda\}\leq \frac{4}{\sqrt{\pi a}}\exp\left\{-\frac{a^2}{8K\lambda}\right\}$$ holds for the moment of time $\tau_a$ at which semimartingale $X$ leaves the interval $[X_0-a, X_0+a]$ for the first time. Consider the Ornstein--Uhlenbeck process $Y_t=Y_0-\alpha \int_0^tY_sds+k W_t,$  where $W$ is a Wiener process, and choose arbitrary $N>a+|Y_0|.$ Denote $\tau^N=\inf\{t>0: |Y_t|\geq N\}$. Then $\tau_N>\tau_a.$ Furthermore
$$\widehat{Y}_t:=Y_{t\wedge\tau_N}=Y_0-\alpha \int_0^tY_{s\wedge\tau_N}1_{s\leq \tau_N}ds+k \int_0^t 1_{s\leq \tau_N} dW_s,$$
and semimartingale $\widehat{Y}_t$ satisfies the conditions of Lemma 10.5 \cite{Watanabe} for $K=N\cdot(\alpha\vee k)$. Moreover, denote
 $\tau_a^N=\inf\{t>0: |\widetilde{Y}_t-Y_0|\geq a\}$ and notice that $\tau_a^N=\tau_a$. Then in the vicinity of zero the distribution of the moment of time $\tau$ admits exponential estimate: there exist constants $C_1, C_2, C_3$ such that $$\mathbb{ P}\{\tau_a<\lambda\}\leq C_1\exp\left\{-\frac{C_2}{\lambda}\right\}, 0<\lambda<C_3,$$
 and the same holds for  $\tau_1.$ This proves Lemma for the case of Ornstein--Uhlenbeck process. The case of Cox--Ingersoll--Ross is dealt with absolutely similarly. Lemma is proved.
\end{proof}

Now we prove some technical results concerning the form of stochastic derivatives and estimates for them.
\begin{lem}\label{Dheta}
Let the conditions from item 1) of Theorem \ref{TheoremDensity} are fulfilled. Then stochastic derivative $D_h\eta_{t}$ has the form

\begin{equation}\begin{gathered}\label{deriv}
D_h\eta_{t}=\frac{\alpha T}{k}\operatorname{e}^{-\alpha t}\Bigg(\operatorname{e}^{-\alpha (t-h)}1_{\{h<t\}}\nu'(Y_t)
 \Big[ \int\limits_{0}^T \int\limits_{0}^T\Big[\operatorname{e}^{-\alpha |t_1-t_2|}
 -\operatorname{e}^{-\alpha (t_1+t_2)}\Big]\\
 \times\nu(Y_{t_1}) \nu(Y_{t_2})\,dt_1dt_2\Big]^{-1}-\nu(Y_t)
\Big[\int\limits_{0}^T \int\limits_{0}^T\Big[\operatorname{e}^{-\alpha |t_1-t_2|}
-\operatorname{e}^{-\alpha (t_1+t_2)}\Big]\\
\times\nu(Y_{t_1})\nu(Y_{t_2}) \,dt_1dt_2\Big]^{-2}
\int\limits_{0}^T \int\limits_{0}^T\Big[\operatorname{e}^{-\alpha |t_1-t_2|}
-\operatorname{e}^{-\alpha (t_1+t_2)}\Big]\\
\times\left(\nu(Y_{t_1})\operatorname{e}^{-\alpha (t_2-h)}1_{\{h<t_2\}}\nu'(Y_{t_2})+\nu(Y_{t_2})\operatorname{e}^{-\alpha
(t_1-h)}1_{\{h<t_1\}}\nu'(Y_{t_1})\right)\,dt_1dt_2\Bigg).
\end{gathered}\end{equation}
\end{lem}
\begin{proof}
It worth notice that in order to be completely correct we need to begin with checking that double integral in denominator of the expression for $\eta_{t}$ and the expression for stochastic derivative of this function is almost surely positive. However we perform this check in Lemma \ref{IntSkor}, where the integrability of stochastic derivative is proven, and here we limit ourselves to derivation of its form.

By the chain rule
\begin{multline*}
D_h\eta_{t}={\alpha T\over k}\operatorname{e}^{-\alpha t}\\
\times D_h \left(\nu(Y_t)
\left[ \int\limits_{0}^T \int\limits_{0}^T\Big[\operatorname{e}^{-\alpha |t_1-t_2|}-\operatorname{e}^{-\alpha (t_1+t_2)}\Big]\nu(Y_{t_1})\nu(Y_{t_2})
\,dt_1dt_2\right]^{-1}\right)
\end{multline*}
\begin{multline*}
={\alpha T\over k}\operatorname{e}^{-\alpha t}\left\{D_h\nu(Y_t) \left[ \int\limits_{0}^T \int\limits_{0}^T\Big[\operatorname{e}^{-\alpha
|t_1-t_2|}-\operatorname{e}^{-\alpha (t_1+t_2)}\Big]\nu(Y_{t_1})\nu(Y_{t_2}) \,dt_1dt_2\right]^{-1}\right.\\
\left.\quad+\nu(Y_t)D_h\left(\left[ \int\limits_{0}^T \int\limits_{0}^T\Big[\operatorname{e}^{-\alpha |t_1-t_2|}-\operatorname{e}^{-\alpha
(t_1+t_2)}\Big]\nu(Y_{t_1})\nu(Y_{t_2}) \,dt_1dt_2\right]^{-1}\right)\right\}
\end{multline*}
\begin{multline*}
={\alpha T\over k}\operatorname{e}^{-\alpha t}\Bigg\{D_hY_t\nu'(Y_t)\left[ \int\limits_{0}^T \int\limits_{0}^T\Big[\operatorname{e}^{-\alpha
|t_1-t_2|}-\operatorname{e}^{-\alpha (t_1+t_2)}\Big]
\nu(Y_{t_1})\nu(Y_{t_2})\,dt_1dt_2\Bigg]^{-1}\right.\\
-\nu(Y_t)\left[ \int\limits_{0}^T \int\limits_{0}^T\Big[\operatorname{e}^{-\alpha |t_1-t_2|}-\operatorname{e}^{-\alpha
(t_1+t_2)}\Big]\nu(Y_{t_1})\nu(Y_{t_2})\,dt_1dt_2\right]^{-2}\int\limits_{0}^T \int\limits_{0}^T\Big[\operatorname{e}^{-\alpha |t_1-t_2|}-\\
\operatorname{e}^{-\alpha (t_1+t_2)}\Big]\Big[\left(\nu(Y_{t_1})D_hY_{t_2}\nu'(Y_{t_2})+\nu(Y_{t_2})D_hY_{t_1}\nu'(Y_{t_1})\right)\Big]\,dt_1dt_2\Bigg\}.
\end{multline*}

As $D_hY_t=\operatorname{e}^{-\alpha (t-h)}1_{\{h<t\}}$ the lema is proved.
\end{proof}

\begin{lem}\label{IntSkor}
Let the conditions from item 1) of Theorem \ref{TheoremDensity} are fulfilled. Then $\eta_t  \in L^{1,2}$ and as a consequence for any $h\in [0,T]:$  $\eta_t 1_{\{h<t\}} \in
L^{1,2}.$
\end{lem}
\begin{proof}
In order to prove the statement of the lemma we need to check that inequality
\begin{equation*}\begin{gathered}\label{est1}
||\eta_t ||^2_{L^{1,2}}=\mathbb{E}\left[\int_0^T\eta^2_t dt
+\int_0^T\int_0^T \left(D_h \eta_t \right)^2 dtdh\right]< \infty
\end{gathered}\end{equation*}
holds.
Let us show that the first summand from the right-hand side of the above expression is bounded.

We introduce the set $A=\{ x \in \mathbb{R}: |\sigma'(x)-\sigma'(Y_0)|\geq \sigma'(Y_0)/2\}$. Recall that $\tau=\inf\left\{t: Y_t \in A\right\},$ $\tau_1=\tau
\wedge T.$ The fact that the expression
$$
\Big[\operatorname{e}^{-\alpha |t_1-t_2|}-\operatorname{e}^{-\alpha (t_1+t_2)}\Big]\nu(Y_{t_1})\nu(Y_{t_2})
$$
 is nonnegative by the conditions of the theorem along with the assumption (A2) which gives that $\sigma(x) \geq c >0$ for $x \in \mathds{R}$ and some constant $c>0,$  yield the inequality $\nu(Y_t)\geq
 \frac{c}{2}\sigma'(Y_0)$ for $t \in (0,\tau).$  Then
\begin{equation}\begin{gathered}\label{equat1}
\eta_{t}={\alpha T\over k}\operatorname{e}^{-\alpha t} \nu(Y_t) \\
\times\left( \int\limits_{0}^T \int\limits_{0}^T\Big(\operatorname{e}^{-\alpha |t_1-t_2|}-\operatorname{e}^{-\alpha (t_1+t_2)}\Big)\nu(Y_{t_1})\nu(Y_{t_2})
\,dt_1dt_2\right)^{-1}\\
\leq {\alpha T\over k}\operatorname{e}^{-\alpha t} \nu(Y_t) \\
\times\left[ \int\limits_{0}^{\tau_1} \int\limits_{0}^{\tau_1}\Big[\operatorname{e}^{-\alpha |t_1-t_2|}-\operatorname{e}^{-\alpha
(t_1+t_2)}\Big]\nu(Y_{t_1})\nu(Y_{t_2}) \,dt_1dt_2\right]^{-1}\\
\leq \frac{4 \alpha T}{k(c\sigma'(Y_0))^2}\operatorname{e}^{-\alpha t} \nu(Y_t)\left[ \int\limits_{0}^{\tau_1}
\int\limits_{0}^{\tau_1}\Big(\operatorname{e}^{-\alpha
|t_1-t_2|}-\operatorname{e}^{-\alpha (t_1+t_2)}\Big)\,dt_1dt_2\right]^{-1}.
\end{gathered}\end{equation}

Below we denote by $C$ or  $C$ with indexes constants values of which are unimportant. Consider the double integral in the denominator and evaluate it in arbitrary point $x>0$:
\begin{gather*}
\psi(x):=\int_{0}^{x} \int_{0}^{x}\Big(\operatorname{e}^{-\alpha |t_1-t_2|}-\operatorname{e}^{-\alpha (t_1+t_2)}\Big)\,dt_1dt_2\\=
\frac{\alpha}{2}\int_0^{x} \Big(\int_h^{x}\operatorname{e}^{-\alpha (s-h)}
\,ds\Big)^2 dh \\
=\frac{1}{2\alpha}\int_0^{x}\big(\operatorname{e}^{2\alpha (h-x)}-2\operatorname{e}^{\alpha
(h-x)}+1\big)dh\\=\frac{1}{2\alpha}\left(\frac{-\operatorname{e}^{-2\alpha x}+1}{2\alpha}+2\frac{\operatorname{e}^{-\alpha
x}-1}{\alpha}+x\right)\\
=C(4\operatorname{e}^{-\alpha x}-\operatorname{e}^{-2\alpha x}+2\alpha x-3).
\end{gather*}

Notice that  $\psi(0)=0$ and the function $\psi(x)=C(4\operatorname{e}^{-\alpha x}-\operatorname{e}^{-2\alpha x}+2\alpha x-3)/(4\alpha^2)$ increases on $\mathbb{R}$.
We make use of elementary inequality $1-\operatorname{e}^{-\alpha x} \geq \alpha x \operatorname{e}^{-\alpha x},$   $x \geq 0,$ in order to find the lower bound for the derivative of the function $\psi:$
\begin{equation*}
\psi'(x)=\frac{1}{2\alpha}(1-\operatorname{e}^{-\alpha x})^2 \geq Cx^2 \operatorname{e}^{-2\alpha x}.
\end{equation*}
Hence
\begin{equation*}
\psi(x)=\int_0^{x}\psi'(s)ds \geq C\int_0^{x}  {e}^{-2\alpha s} s^2 ds > C_1x^3.
\end{equation*}
Then with probability one the following inequality holds:
 \begin{equation}\begin{gathered}\label{tau1}\int\limits_{0}^T \int\limits_{0}^T\Big(\operatorname{e}^{-\alpha |t_1-t_2|}-\operatorname{e}^{-\alpha (t_1+t_2)}\Big)\nu(Y_{t_1})\nu(Y_{t_2})\,dt_1dt_2\geq C \tau_1^3\end{gathered}\end{equation}
Moreover   $\nu^2(x) \leq C(1+|x|^m)$ for some $C>0,$ $m \in \mathds{N},$ because $\sigma(x)$ and $\sigma'(x)$ are of at most polynomial growth. Taking into account \eqref{equat1} we arrive at inequality
$$\eta_{t}\leq C\operatorname{e}^{-\alpha t} \nu(Y_t)\tau_1^{-3}\leq C(1+|Y_t|^m)\tau_1^{-3}.$$ The moments of any order of Ornstein--Uhlenbeck process are uniformly bounded on any interval, so taking into account Lemma \ref{lemlem} we get
\begin{equation}\label{est21}  \sup_{t\in T}\mathbb{E}\eta^2_t \leq C\sup_{t\in T}(\mathbb{E}(1+|Y_t|^{2m})(\mathbb{E}\tau_1^{-6}))^{\frac{1}{2}}\leq C.
\end{equation}
This yields the necessary estimate
\begin{equation}\label{Ref23}
\begin{gathered}
\mathbb{E}\left[\int_0^T \eta^2_t dt\right]
<\infty.
\end{gathered}
\end{equation}

Now we consider the second summand in the norm. First,using \eqref{deriv} and \eqref{tau1} we assess the stochastic derivative:
\begin{equation}\label{IntDheta}
\begin{gathered}
   |\left(D_h \eta_{t} \right)|
\leq C\left(\tau_1^{-3}\nu'(Y_t) + \tau_1^{-6}\nu'(Y_t)  \int_0^T\nu(Y_s)ds\int_0^T\nu'(Y_u)du \right)\\ \leq C(1+|Y_t|^{m_1})(\tau_1^{-3}+\tau_1^{-6})
\end{gathered}
\end{equation}
for some $m_1\in \mathds{N}$. Considerations similar to those covered in the proof of \eqref{est21} yield that there exists a constant $C>0$ such that
\begin{equation}\label{est23} \sup_{t\in T}\mathbb{E} \left(D_h \eta_{t} \right)^2\leq C\;\; \text{i}\;\;\int_0^T\int_0^T \left(D_h \eta_t \right)^2 dtdh<\infty,
\end{equation} and by definition $\eta_t  \in L^{1,2}.$ Lemma is proved.
\end{proof}

\begin{lem}\label{ZetaIntSkor}
Let the requirements of item 1) of Theorem \ref{TheoremDensity} are fulfilled. Then $\overline{\zeta}_h \in L^{1,2}.$
\end{lem}
\begin{proof}
In order to prove the statement of the lemma we need to prove that the following inequality holds:
\begin{equation}\label{ZetaL12}
\begin{gathered}
||\overline{\zeta}_h||^2_{L^{1,2}}=\mathbb{E}\Bigg[\int_0^T\left(\operatorname{e}^{\alpha h}\int_h^T \eta_{s}\, ds\right)^2dh\Bigg]\\
+\mathbb{E}\Bigg[\int_0^T\int_0^T \left(D_h \left(\operatorname{e}^{\alpha t}\int_t^T \eta_{s}\, ds\right)\right)^2 dtdh \Bigg] < \infty.
\end{gathered}
\end{equation}
For the first summand in \eqref{ZetaL12} by  \label{Ref2} we have

\begin{gather*}
\mathbb{E}\Bigg[\int_0^T\left(\operatorname{e}^{\alpha h}\int_h^T \eta_{s}\, ds\right)^2dh\Bigg] \leq
T\int_0^T\operatorname{e}^{2\alpha h}\mathbb{E}\Bigg[\int_0^T \eta^2_{s}\, ds\Bigg]dh < \infty,
\end{gather*}
and according to \eqref{est23} for the second summand:
\begin{gather*}
\mathbb{E}\Bigg[\int_0^T\int_0^T \left(D_h \left(\operatorname{e}^{\alpha t}\int_t^T \eta_{s}\, ds\right)\right)^2 dtdh\Bigg] \\
=\int_0^T\int_0^T \operatorname{e}^{2\alpha t}\mathbb{E}\Bigg[\left(\int_0^T D_h(\eta_{s}1_{\{t<s\}})\, ds\right)^2 \Bigg]dtdh \\
\leq T\int_0^T\int_0^T \operatorname{e}^{2\alpha t}\mathbb{E}\Bigg[\int_0^T (D_h\eta_{s})^2\, ds \Bigg]dtdh < \infty.
\end{gather*}
\end{proof}

\begin{zau}\label{moment}
In order to prove the following lemma we recall the result from \cite[inequality (3.1)]{Andreas}: $\sup\limits_{[0,T]}\mathbb{E} Z_t^p<\infty$ for any
$p>-\frac{2b}{k^2}$.
\end{zau}

\begin{lem}
\label{lem_norma} Let the coefficients of Cox--Ingersoll--Ross process given by equation \eqref{ModelB1} satisfy the inequality $6k^2<b$. Then  $\sqrt{Z_t}\Psi_{h,t}  \in L^{1,2}$ and   $\int_h^{T}\sqrt{Z_t}\Psi_{h,t}dt  \in L^{1,2}$.

\end{lem}
\begin{proof}
We prove the second statement of the lemma which requires more transforms. The proof of first statement is similar.It is necessary to show that
\begin{equation}\begin{gathered}\label{eq_norma}
\left\|\int_h^{T}\sqrt{Z_t}\Psi_{h,t}dt\right\|^2_{L^{1,2}}=\mathbb{E}\left[\int_0^T\left(\int_h^{T}\sqrt{Z_t}\Psi_{h,t}dt\right)^2dh\right.\\
\left.+\int_0^T\int_0^T \left(D_l \int_h^{T}\sqrt{Z_t}\Psi_{h,t}dt\right)^2 dldh\right]< \infty.
\end{gathered}\end{equation}
Consider the function
$$
\sqrt{Z_t}\Psi_{h,t}=\sqrt{Z_t}\psi_{h,t}\left[\int_{0}^T\int_{0}^T \sqrt{Z_{t_1}}\sqrt{Z_{t_2}}\int_0^{t_1\wedge
t_2}\psi_{h,t_1}\psi_{h,t_2}dh\,dt_1dt_2\right]^{-1}.\\
$$
 We need to find an estimate for the lower bound of the denominator of this function. Denote
\begin{multline*}
I:=\int_{0}^T\int_{0}^T \sqrt{Z_{t_1}}\sqrt{Z_{t_2}}\int_0^{t_1\wedge t_2}\psi_{h,t_1}\psi_{h,t_2}dh\,dt_1dt_2\\
=\int_{0}^T\int_{0}^T \sqrt{Z_{t_1}}\sqrt{Z_{t_2}}\int_0^{t_1\wedge
t_2}\exp\left\{-\frac{t_1-h}{2}-\left(\frac{b}{2}-\frac{k^2}{8}\right)\int_h^{t_1}\frac{ds}{Z_s}\right\}\\
\times\exp\left\{-\frac{t_2-h}{2}-\left(\frac{b}{2}-\frac{k^2}{8}\right)\int_h^{t_2}\frac{ds}{Z_s}\right\}dh\,dt_1dt_2.
\end{multline*}

Recall that $\widetilde{\tau}:=\inf\left\{t:|Z_t-Z_0|>\frac{Z_0}{2}\right\}$, $\widetilde{\tau}_1=\widetilde{\tau}\wedge T$. Denote by $q:=\frac{b}{2}-\frac{k^2}{8}$. Conditions of the lemma yield that $q>0$. We get
\begin{multline*}
I\ge\frac{Z_0}{2}\int_{0}^{\widetilde{\tau}_1}\int_{0}^{\widetilde{\tau}_1} \int_0^{t_1\wedge
t_2}\exp\left\{-\frac{t_1-h}{2}-\frac{2q}{Z_0}\int_h^{t_1}ds\right\}\\
\times\exp\left\{-\frac{t_2-h}{2}-\frac{2q}{Z_0}\int_h^{t_2}ds\right\}dh\,dt_1dt_2
\end{multline*}
\begin{multline*}
=\frac{Z_0}{2}\int_{0}^{\widetilde{\tau}_1}\int_{0}^{\widetilde{\tau}_1} \int_0^{t_1\wedge t_2}\exp\left\{-\frac{t_1-h}{2}-\frac{2q}{Z_0}(t_1-h)\right\}\\
\times\exp\left\{-\frac{t_2-h}{2}-\frac{2q}{Z_0}(t_2-h)\right\}dh\,dt_1dt_2.
\end{multline*}
As $t_1-h<T$ and $t_2-h<T$ the following inequality holds
$$
I\ge\frac{Z_0}{12}\exp\left\{-\frac{T(Z_0+4q)}{Z_0}\right\}\widetilde{\tau}_1^3.
$$
Below we denote unimportant constants by $C$ or $C$ with indexes. Notice that the function $\psi_{h,t}$ is bounded from above by 1. Hence
$$
I\ge C\widetilde{\tau}_1^3, \;\;\text{i}\;\; \Psi_{h,t}\leq C\widetilde{\tau}_1^{-3}. $$

  Then $$
\mathbb{E}({Z_t}\Psi^2_{h,t})\leq \left(\mathbb{E}Z_t^2\mathbb{E}\Psi^2_{h,t}\right)^ {\frac{1}{2}}\leq C\left(\mathbb{E}Z_t^2\right)^{\frac{1}{2}}\left(\mathbb{E}\widetilde{\tau}_1^{-6}\right)^{\frac{1}{2}}.$$
The conditions of the lemma and remark \ref{moment} provide that  $\sup_{t\in[0,T]}\mathbb{E}Z_t^p<\infty$ for any
$p\ge -12$. Boundedness of $\mathbb{E}\widetilde{\tau}_1^{-6}$ is provided by Lemma \ref{lemlem}. Then there exists a constant $C>0$ such that
\begin{equation}
\label{sup1}\sup_{t,h\in[0,T]}\mathbb{E}({Z_t}\Psi^2_{h,t})\leq C.
\end{equation}
The first summand from the right-hand side of equality \eqref{eq_norma} is bounded. Really, \begin{equation*}\begin{gathered}
\mathbb{E}\left[\int_0^T\left(\int_h^{T}\sqrt{Z_t}\Psi_{h,t}dt\right)^2dh\right]
\le T    \int_0^T \int_h^{T}\mathbb{E} (Z_t \Psi^2_{h,t})dt  dh \leq C_1.
\end{gathered}\end{equation*}

Let us show that the second summand from the right-hand side of equality \eqref{eq_norma} is bounded/
\begin{equation}\begin{gathered}\label{eq_norma_2}
\mathbb{E}\left[\int_0^T\int_0^T \left(D_l \int_h^{T}\sqrt{Z_t}\Psi_{h,t}dt\right)^2 dldh\right]\\
=\mathbb{E}\left[\int_0^T\int_0^T \left( \int_h^{T}D_l(\sqrt{Z_t}\Psi_{h,t})dt\right)^2\right]dldh\\
\le T\int_0^T\int_0^T   \int_h^{T}\mathbb{E}(D_l(\sqrt{Z_t}\Psi_{h,t}))^2 dt\,  dl\,dh.
\end{gathered}\end{equation}
It is necessary to find the expression for the stochastic derivative. Using the chain rule we get:
\begin{equation*}\begin{gathered}
D_l(\sqrt{Z_t}\Psi_{h,t})=D_l\left(\sqrt{Z_t}\psi_{h,t}\left[\int_{0}^T\int_{0}^T \sqrt{Z_{t_1}}\sqrt{Z_{t_2}}\int_0^{t_1\wedge
t_2}\psi_{h,t_1}\psi_{h,t_2}dh\,dt_1dt_2\right]^{-1}\right)
\end{gathered}\end{equation*}
\begin{multline*}
=D_l\left(\sqrt{Z_t}\exp\left\{-\frac{t-h}{2}-q\int_h^t\frac{ds}{Z_s}\right\}\times\right.\\
\left.\times\left[\int_{0}^T\int_{0}^T \int_0^{t_1\wedge
t_2}\sqrt{Z_{t_1}}\sqrt{Z_{t_2}}\exp\left\{-\frac{t_1-h}{2}-q\int_h^{t_1}\frac{ds}{Z_s}\right\}\times\right.\right.\\
\left.\left.\times\exp\left\{-\frac{t_2-h}{2}-q\int_h^{t_2}\frac{ds}{Z_s}\right\}dh\,dt_1dt_2\right]^{-1}\right)\\
 =D_l\left(\sqrt{Z_t}\exp\left\{-\frac{t-h}{2}-q\int_h^t\frac{ds}{Z_s}\right\}\right)\\
\times\left[\int_{0}^T\int_{0}^T \int_0^{t_1\wedge
t_2}\sqrt{Z_{t_1}}\sqrt{Z_{t_2}}\exp\left\{-\frac{t_1-h}{2}-q\int_h^{t_1}\frac{ds}{Z_s}\right\}\times\right.\\
\left.\times\exp\left\{-\frac{t_2-h}{2}-q\int_h^{t_2}\frac{ds}{Z_s}\right\}dh\,dt_1dt_2\right]^{-1}
\end{multline*}
\begin{multline*}
+\sqrt{Z_t}\exp\left\{-\frac{t-h}{2}-q\int_h^t\frac{ds}{Z_s}\right\}\\
\times D_l\left[\int_{0}^T\int_{0}^T \int_0^{t_1\wedge
t_2}\sqrt{Z_{t_1}}\sqrt{Z_{t_2}}\exp\left\{-\frac{t_1-h}{2}-q\int_h^{t_1}\frac{ds}{Z_s}\right\}\times\right.\\
\left.\times\exp\left\{-\frac{t_2-h}{2}-q\int_h^{t_2}\frac{ds}{Z_s}\right\}dh\,dt_1dt_2\right]^{-1}
\end{multline*}
\begin{multline*}
=\exp\left\{-\frac{t-h}{2}-q\int_h^t\frac{ds}{Z_s}\right\}\left(\frac{D_l Z_t}{2\sqrt{Z_t}}+q\sqrt{Z_t}\int_h^t\frac{D_l Z_s}{Z_s^2}ds\right)\\
\times\left[\int_{0}^T\int_{0}^T \int_0^{t_1\wedge
t_2}\sqrt{Z_{t_1}}\sqrt{Z_{t_2}}\exp\left\{-\frac{t_1-h}{2}-q\int_h^{t_1}\frac{ds}{Z_s}\right\}\times\right.\\
\left.\times\exp\left\{-\frac{t_2-h}{2}-q\int_h^{t_2}\frac{ds}{Z_s}\right\}dh\,dt_1dt_2\right]^{-1}
\end{multline*}
\begin{multline*}
-\sqrt{Z_t}\exp\left\{-\frac{t-h}{2}-q\int_h^t\frac{ds}{Z_s}\right\}\\
\times\left[\int_{0}^T\int_{0}^T \int_0^{t_1\wedge
t_2}\sqrt{Z_{t_1}}\sqrt{Z_{t_2}}\exp\left\{-\frac{t_1-h}{2}-q\int_h^{t_1}\frac{ds}{Z_s}\right\}\times\right.\\
\left.\times\exp\left\{-\frac{t_2-h}{2}-q\int_h^{t_2}\frac{ds}{Z_s}\right\}dh\,dt_1dt_2\right]^{-2}\times
\end{multline*}
\begin{multline*}
\times\int_{0}^T\int_{0}^T \int_0^{t_1\wedge
t_2}\exp\left\{-\frac{t_1-h}{2}-q\int_h^{t_1}\frac{ds}{Z_s}\right\}\exp\left\{-\frac{t_2-h}{2}-q\int_h^{t_2}\frac{ds}{Z_s}\right\}\times\\
\times \left(\sqrt{Z_{t_2}}\left(\frac{D_lZ_{t_1}}{2\sqrt{Z_{t_1}}}+q\sqrt{Z_{t_1}}\int_h^{t_1}\frac{D_lZ_{s}}{Z_s^2}ds\right)+\right.\\
\left.+\sqrt{Z_{t_1}}\left(\frac{D_lZ_{t_2}}{2\sqrt{Z_{t_2}}}+q\sqrt{Z_{t_2}}\int_h^{t_2}\frac{D_lZ_{s}}{Z_s^2}ds\right) \right)dh\,dt_1dt_2.
\end{multline*}
Taking into account the form of $D_l Z_t$, $\psi_{h,t}$ and $I$ we arrive at
\begin{multline*}
D_l(\sqrt{Z_t}\Psi_{h,t})=k\psi_{h,t}\left(\frac{\psi_{l,t}}{2}+q\sqrt{Z_t}\int_h^t\frac{\psi_{l,s}}{Z_s^{\frac{3}{2}}}ds\right)I^{-1}\\
-k\sqrt{Z_t}\psi_{h,t}I^{-2}\int_{0}^T\int_{0}^T \int_0^{t_1\wedge
t_2}\psi_{h,t_1}\psi_{h,t_2}\left(\sqrt{Z_{t_2}}\left(\frac{\psi_{l,t_1}}{2}+q\sqrt{Z_{t_1}}\int_h^{t_1}\frac{\psi_{l,s}}{Z_s^{\frac{3}{2}}}ds\right)+\right.\\
\left.+\sqrt{Z_{t_1}}\left(\frac{\psi_{l,t_2}}{2}+q\sqrt{Z_{t_2}}\int_h^{t_2}\frac{\psi_{l,s}}{Z_s^2}ds\right) \right)dh\,dt_1dt_2.
\end{multline*}
Taking into account \eqref{eq_norma_2} in order to prove the boundedness of the second summand in \eqref{eq_norma}it suffices to show that
\begin{equation}\label{sup2}\sup\limits_{l,h,t}
\mathbb{E}\left[(D_l(\sqrt{Z_t}\Psi_{h,t}))^2\right]<\infty.
\end{equation}
To this end we consider
\begin{multline*}
\mathbb{E}\left[(D_l(\sqrt{Z_t}\Psi_{h,t}))^2\right]=\mathbb{E}\left[\left(k\psi_{h,t}\left(\frac{\psi_{l,t}}{2}+q\sqrt{Z_t}\int_h^t\frac{\psi_{l,s}}{Z_s^{\frac{3}{2}}}ds\right)I^{-1}\right.\right.\\
-k\sqrt{Z_t}\psi_{h,t}I^{-2}\int_{0}^T\int_{0}^T \int_0^{t_1\wedge
t_2}\psi_{h,t_1}\psi_{h,t_2}\left(\sqrt{Z_{t_2}}\left(\frac{\psi_{l,t_1}}{2}+q\sqrt{Z_{t_1}}\int_h^{t_1}\frac{\psi_{l,s}}{Z_s^{\frac{3}{2}}}ds\right)+\right.\\
\left.\left.\left.+\sqrt{Z_{t_1}}\left(\frac{\psi_{l,t_2}}{2}+q\sqrt{Z_{t_2}}\int_h^{t_2}\frac{\psi_{l,s}}{Z_s^2}ds\right) \right)dh\,dt_1dt_2\right)^2\right]
\end{multline*}
\begin{multline*}
\le2k^2\mathbb{E}\left[\left(\psi_{h,t}\left(\frac{\psi_{l,t}}{2}+q\sqrt{Z_t}\int_h^t\frac{\psi_{l,s}}{Z_s^{\frac{3}{2}}}ds\right)I^{-1}\right)^2\right]\\
+2k^2\mathbb{E}\left[\left(\sqrt{Z_t}\psi_{h,t}I^{-2}\int_{0}^T\int_{0}^T \int_0^{t_1\wedge
t_2}\psi_{h,t_1}\psi_{h,t_2}\left(\sqrt{Z_{t_2}}\left(\frac{\psi_{l,t_1}}{2}+q\sqrt{Z_{t_1}}\int_h^{t_1}\frac{\psi_{l,s}}{Z_s^{\frac{3}{2}}}ds\right)+\right.\right.\right.\\
\left.\left.\left.+\sqrt{Z_{t_1}}\left(\frac{\psi_{l,t_2}}{2}+q\sqrt{Z_{t_2}}\int_h^{t_2}\frac{\psi_{l,s}}{Z_s^2}ds\right) \right)dh\,dt_1dt_2\right)^2\right].
\end{multline*}
Hence, provided by the inequality $\psi_{h,t}<1$, we get:
\begin{multline*}\label{eq_33}
\mathbb{E}\left[(D_l(\sqrt{Z_t}\Psi_{h,t}))^2\right] \le
2k^2\mathbb{E}\left[\left(\frac{1}{2}+q\sqrt{Z_t}\int_h^t\frac{ds}{Z_s^{\frac{3}{2}}}\right)^2I^{-2}\right]\\
+2k^2\mathbb{E}\left[\left(\sqrt{Z_t}I^{-2}\int_{0}^T\int_{0}^T \int_0^{t_1\wedge
t_2}\left(\sqrt{Z_{t_2}}\left(\frac{1}{2}+q\sqrt{Z_{t_1}}\int_h^{t_1}\frac{ds}{Z_s^{\frac{3}{2}}}\right)+\right.\right.\right.\\
\left.\left.\left.+\sqrt{Z_{t_1}}\left(\frac{1}{2}+q\sqrt{Z_{t_2}}\int_h^{t_2}\frac{ds}{Z_s^2}\right) \right)dh\,dt_1dt_2 \right)^2 \right]=2k^2I_1+2k^2I_2.
\end{multline*}
We estimate each  expectation separately. By H\"{o}lder's inequality and the inequality $(a+b)^n\le 2^{n-1}(a^n+b^n)$ we get
\begin{multline*}
I_1:=\mathbb{E}\left[\left(\frac{1}{2}+q\sqrt{Z_t}\int_h^t\frac{ds}{Z_s^{\frac{3}{2}}}\right)^2 I^{-2}\right]\le
\left(\mathbb{E}\left[\left(\frac{1}{2}+q\sqrt{Z_t}\int_h^t\frac{ds}{Z_s^{\frac{3}{2}}}\right)^4\right]\right)^{1/2}\left(\mathbb{E}I^{-4}\right)^{1/2}\\
\le
C_1\left(\frac{1}{2}+8q^4\mathbb{E}\left[\left(\sqrt{Z_t}\int_h^t\frac{ds}{Z_s^{\frac{3}{2}}}\right)^4\right]\right)^{\frac{1}{2}}\left(\mathbb{E}\widetilde{\tau}_1^{-12}\right)^{\frac{1}{2}}.
\end{multline*}
The boundedness of $\mathbb{E}\widetilde{\tau}_1^{-12}$ is provided by Lemma  \ref{lemlem}. Now we need to estimate

\begin{gather*}
\mathbb{E}\left[\left(\sqrt{Z_t}\int_h^t\frac{ds}{Z_s^{\frac{3}{2}}}\right)^4\right]\le\left(\mathbb{E}Z_t^4\right)^\frac{1}{2}\left(\mathbb{E}\left[\left(\int_h^t\frac{ds}{Z_s^{\frac{3}{2}}}\right)^8\right]\right)^\frac{1}{2}\\
\le T^{\frac{7}{2}}\left(\mathbb{E}Z_t^4\right)^\frac{1}{2}\left(\int_h^t\mathbb{E}Z_s^{-12}ds\right)^\frac{1}{2},
\end{gather*}
Taking to account Remark \ref{moment} and the conditions of the lemma we arrive at
$$\mathbb{E}\left[\left(\sqrt{Z_t}\int_h^t\frac{ds}{Z_s^{\frac{3}{2}}}\right)^4\right]\leq C.$$ And this means that $I_1\leq C$.
The next estimate is for $I_2$.
\begin{multline*}
I_2:=\mathbb{E}\left[\left(\sqrt{Z_t}I^{-2}\int_{0}^T\int_{0}^T \int_0^{t_1\wedge
t_2}\left(\sqrt{Z_{t_2}}\left(\frac{1}{2}+q\sqrt{Z_{t_1}}\int_h^{t_1}\frac{ds}{Z_s^{\frac{3}{2}}}\right)+\right.\right.\right.\\
\left.\left.\left.+\sqrt{Z_{t_1}}\left(\frac{1}{2}+q\sqrt{Z_{t_2}}\int_h^{t_2}\frac{ds}{Z_s^2}\right) \right)dh\,dt_1dt_2\right)^2\right]
\end{multline*}
\begin{multline*}
\le\left(\mathbb{E}(Z_t^2I^{-8})\right)^{\frac{1}{2}}\left(\mathbb{E}\left[\left(\int_{0}^T\int_{0}^T \int_0^{t_1\wedge
t_2}\left(\sqrt{Z_{t_2}}\left(\frac{1}{2}+q\sqrt{Z_{t_1}}\int_h^{t_1}\frac{ds}{Z_s^{\frac{3}{2}}}\right)+\right.\right.\right.\right.\\
\left.\left.\left.\left.+\sqrt{Z_{t_1}}\left(\frac{1}{2}+q\sqrt{Z_{t_2}}\int_h^{t_2}\frac{ds}{Z_s^2}\right)
\right)dh\,dt_1dt_2\right)^4\right]\right)^{\frac{1}{2}}
\end{multline*}
\begin{multline*}
\le T^{\frac{9}{2}}\left(\mathbb{E}Z_t^4\right)^\frac{1}{4}\left(\mathbb{E}I^{-16}\right)^\frac{1}{4}\left(\int_{0}^T\int_{0}^T \int_0^{t_1\wedge
t_2}\mathbb{E}\left[\left(\sqrt{Z_{t_2}}\left(\frac{1}{2}+q\sqrt{Z_{t_1}}\int_h^{t_1}\frac{ds}{Z_s^{\frac{3}{2}}}\right)+\right.\right.\right.\\
\left.\left.\left.+\sqrt{Z_{t_1}}\left(\frac{1}{2}+q\sqrt{Z_{t_2}}\int_h^{t_2}\frac{ds}{Z_s^2}\right)\right)^4 \right]dh\,dt_1dt_2\right)^{\frac{1}{2}}
\end{multline*}
\begin{multline*}
\le C_2\left(\mathbb{E}Z_t^4\right)^\frac{1}{4}\left(\mathbb{E}\widetilde{\tau}_1^{-48}\right)^\frac{1}{4}\left(\int_{0}^T\int_{0}^T \int_0^{t_1\wedge
t_2}\mathbb{E}\left[\left(\sqrt{Z_{t_2}}\left(\frac{1}{2}+q\sqrt{Z_{t_1}}\int_h^{t_1}\frac{ds}{Z_s^{\frac{3}{2}}}\right)+\right.\right.\right.\\
\left.\left.\left.+\sqrt{Z_{t_1}}\left(\frac{1}{2}+q\sqrt{Z_{t_2}}\int_h^{t_2}\frac{ds}{Z_s^2}\right)\right)^4 \right]dh\,dt_1dt_2\right)^{\frac{1}{2}},
\end{multline*}
where $C_2=\frac{12^4T^{\frac{9}{2}}}{Z_0^4}\exp\left\{\frac{4T(Z_0+4q)}{Z_0}\right\}$. The boundedness of  $\mathbb{E}\widetilde{\tau}_1^{-48}$ and $\mathbb{E}Z_t^4$  are provided by Remark \ref{moment}, Lemma \ref{lemlem} and the conditions of the lemma. In order to prove the boundedness of \\

$\mathbb{E}\left[\left(\sqrt{Z_{t_2}}\left(\frac{1}{2}+q\sqrt{Z_{t_1}}\int_h^{t_1}\frac{ds}{Z_s^{\frac{3}{2}}}\right)+\sqrt{Z_{t_1}}\left(\frac{1}{2}+q\sqrt{Z_{t_2}}\int_h^{t_2}\frac{ds}{Z_s^2}\right)
\right)^4\right]$
it suffices to show that the following expression is bounded:
\begin{multline*}
\mathbb{E}\left[\left(\sqrt{Z_{t}}\left(\frac{1}{2}+q\sqrt{Z_{t}}\int_h^{t}\frac{ds}{Z_s^{\frac{3}{2}}}\right)\right)^4\right]
\le
\left(\mathbb{E}Z_t^4\right)^{\frac{1}{2}}\left(\mathbb{E}\left[\left(\frac{1}{2}+q\sqrt{Z_{t}}\int_h^{t}\frac{ds}{Z_s^{\frac{3}{2}}}\right)^8\right]\right)^{\frac{1}{2}}\\
\le\left(\mathbb{E}Z_t^4\right)^{\frac{1}{2}}\left(\frac{1}{2}+128\,q^8\,\mathbb{E}\left[\left(\sqrt{Z_{t}}\int_h^{t}\frac{ds}{Z_s^{\frac{3}{2}}}\right)^8\right]\right)^{\frac{1}{2}}.
\end{multline*}
By the assumption of the lemma $6k^2<b$. Then we may choose $1<d<\frac{b}{6k^2}$. Then $-12d>-\frac{2b}{k^2}$ and according to Remark \ref{moment},
$$\sup\limits_{t\in[0,T]}\mathbb{E}Z^{-12d}_t<\infty.$$
Let $p$ is such that $\frac{1}{p}+\frac{1}{d}=1$. We apply H\"{o}lder's inequality twice to arrive at
\begin{equation*}\begin{gathered}
\mathbb{E}\left[\left(\sqrt{Z_{t}}\int_h^{t}\frac{ds}{Z_s^{\frac{3}{2}}}\right)^8\right]\le\left(\mathbb{E}Z_t^{4p}\right)^{\frac{1}{p}}\left(\mathbb{E}\left[\left(\int_h^{t}\frac{ds}{Z_s^{\frac{3}{2}}}\right)^{8d}\right]\right)^{\frac{1}{d}}\\
\le T^{\frac{8d-1}{d}}\left(\mathbb{E}Z_t^{4p}\right)^{\frac{1}{p}}\left(\int_h^{t}\mathbb{E}Z_s^{-12d}ds\right)^{\frac{1}{d}}<\infty.
\end{gathered}\end{equation*}

 Thus, we have shown that $I_2\leq C$. And this means that \eqref{sup2} holds and consequently \eqref{eq_norma} is satisfied.
 Lemma is proved.
\end{proof}

\textbf{The authors are grateful to Oleksii Kulyk for valuable comments and suggestions that have significantly improved the paper and in particular for detailed explanation of application of Malliavin calculus to the problem of finding the form of probability density functions of integral functionals.}
\selectlanguage{english}

\end{document}